\documentclass[11pt]{article}

\title{Sampling with Walsh Transforms}

\author{Yi LU\\
Institute of Software,\\
Chinese Academy of Sciences,\\ 
Beijing, P.R. China\\
\url{luyi666@gmail.com}
}

\date{}

\usepackage{amsmath}
\usepackage{amssymb}
\usepackage{amsthm}

\usepackage{url,cite}
\usepackage{epsfig,graphics,color}

\newtheorem{theorem}{Theorem}
\newtheorem{corollary}{Corollary}

\newtheorem{remark}{Remark}
\newtheorem{proposition}{Proposition}

\begin{document}
\maketitle

\begin{abstract}
With the advent of massive data outputs at a regular rate, admittedly, signal processing technology plays an increasingly key role. Nowadays, signals are not merely restricted to physical sources, they have been extended to digital sources as well. 

Under the general assumption of discrete statistical signal sources, we propose a practical problem of sampling incomplete noisy signals for which we do not know a priori and the sample size is bounded. We approach this sampling problem by Shannon's channel coding theorem. We use an extremal binary channel with high probability of transmission error, which is rare in communication theory. Our main result demonstrates that it is the large Walsh coefficient(s) that characterize(s) discrete statistical signals, regardless of the signal sources. Note that this is a known fact in specific application domains such as images. By the connection of Shannon's theorem, we establish the necessary and sufficient condition for our generic sampling problem for the first time. Finally, we discuss the cryptographic significance of sparse Walsh transform.

\end{abstract}

\section{Introduction}

With the advent of massive data outputs regularly,
we are confronted by the challenge of
big data processing and analysis. Admittedly,
signal processing has become an increasingly key technology. 
An open question is the sampling problem with the signals, for which we assume that we do not know \emph{a priori}.
Due to reasons of practical consideration, sampling is affected
by possibly strong noise and/or the limited measurement precision.
Assuming that the signal source is not restricted to a particular application domain,
we are concerned with a practical and generic problem to sample these nosiy signals.

Our motivation arises from the following problem in modern applied statistics.
Assume the discrete statistical signals in a general setting as follows.
The samples, generated by an arbitrary (possibly noise-corrupted) source $F$,
 are $2^n$-valued for a fixed $n$. 
We assume that the noise source generates
 uniformly-distributed samples\footnote{For the pure digital signal source $F$,
which is our research subject throughout this work,
this assumption is justified by the maximum entropy principle \cite[P278]{it-book}.}.
Note that our assumption on a general setting of
discrete statistical signals is described by the assumption that
$F$ is an arbitrary yet fixed (not necessarily deterministic) function.
It is known to be a hypothesis testing problem
 to test presence of any real signals. 
Traditionally, $F$ is a deterministic function with small or medium input size.
It is computationally easy to collect the \emph{complete and precise} distribution $f$ of $F$.
Based on relative entropy (or Kullback-Leibler distance),
 the conventional approach (aka. the classic distinguisher in statistical cryptanalysis \cite{vaudenay_ccs96,vaudenay_textbook2006})
solves the sampling problem, given the distribution $f$ \emph{a priori}.
Nevertheless, in reality, $F$ might be
 a function that we do not have the complete description, or it might be a
non-deterministic function, or it might just have large input size.
Thus, it is infeasible to collect the complete and precise distribution $f$.
This gives rise to the new generic statistical
sampling problem with discrete incomplete noisy signals, using bounded samples.

In this work, we show that we can solve the generic sampling problem as reliable as possible without knowing \emph{a priori}.
We approach this problem by the novel use of Shannon's channel coding theorem, 
 which establishes the achievability of channel capacity.
This allows to obtain a simple robust solution 
with an arbitrarily small probability of error.
Note that in the conventional approach (i.e., the classic distinguisher), 
the problem statement is slightly different and the solution is of a different form.
Our work uses the binary channel. 
The channel is assumed to have extremely high probability of
transmission error (and we call it the extremal binary channel), 
which is rare in communication theory \cite{Shokrollahi}.
In particular, for the Binary Symmetric Channel (BSC) with crossover probability $(1-d)/2$ and $d$ is small
 (i.e., $|d|\ll 1$), the channel capacity is approximately $d^2/(2\log 2)$. 
Further, we construct
 a non-symmetric binary channel with crossover probability $(1-d)/2$ and $1/2$ respectively (and $d$ is small).
We show that the channel capacity is approximately $d^2/(8\log 2)$.

Our main contributions are as follows.
First, we present the generic sampling theorem.
We
 show that for this extremal non-symmetric binary channel, 
Shannon's channel coding theorem 
can solve the generic sampling problem 
under the general assumption of statistical signal sources (i.e., no further assumption is made about signal sources).
Specifically, the \emph{necessary and sufficient} condition is given \emph{for the first time} 
to sample the incomplete noisy signals with bounded sample size for signal detection.
It is interesting to observe that the classical signal processing tool of Walsh transform
\cite{dsp_book,walsh-book2} is essential: 
regardless of the real signal sources, the large Walsh coefficient(s) 
 characterize(s) discrete statistical signals.
Put other way, when sampling incomplete noisy signals of the same source multiple times,
one can expect to see \emph{repeatedly} those large Walsh coefficient(s) of same magnitude(s) at the
fixed frequency position(s). Note that this is known in specific application
domains such as images, voices etc.
 Clearly,
our result shows strong connection between Shannon's theorem and Walsh transform. 
Both are the key innovative technologies in digital signal processing.

Secondly, our generic sampling theorem is naturally linked to the new area of compressive sensing \cite{cs2006paper}.
Compressive sensing is based on the ground of sparse representation of signals in the transform domain.
This enables powerful sampling techniques (with respect to the complexity of time-domain components for access and the time cost) 
for the purpose of signal recovery.  Specifically,
sparse Fourier transform has been the main research subject in this area.
Most recently, studies on sparse Walsh transform follow \cite{berkeley_sparse_wht_project,epfl_sparse_wht_project}.
Our preliminary work finds that in the most general case, sparse Walsh transform is linked (see \cite{crypto2004me,JoC2008}) to the
 maximum likelihood decoding problem for linear codes, 
which is known to be NP-complete.

The rest of the paper is organized as follows.
In Section \ref{sect_review_statistics}, we give preliminaries on Walsh transforms.
In Section \ref{sect_shannon}, we review Shannon's channel coding theorem.
In Section \ref{sect_stat_tran}, we translate Shannon's theorem in the case of extremal binary channels to 
 hypothesis testing problems. Based on the results,
we present our main sampling theorem in Section \ref{sect_main_result};
we also discuss the cryptographic significance.
We give concluding remarks in Section \ref{sect_end}.

\section{Walsh Transforms in Statistics}
\label{sect_review_statistics}


Given a real-valued function 
$f: GF(2)^n \to \rm{R}$, which is defined on an $n$-tuple binary vector of input,
the Walsh transform of $f$, denoted by $\widehat{f}$, is another real-valued function defined as
\begin{equation}
\widehat{f}(i)=\sum_{j\in GF(2)^n}(-1)^{<i,j>} f(j),
\end{equation}
for all $i \in GF(2)^n$, where $<i,j>$ denotes the inner product between two $n$-tuple binary vectors $i,j$.
For later convenience, we give an alternative definition below.
Given an input array $x=(x_0,x_1,\ldots,x_{2^{n}-1})$ of $2^n$ reals in the time domain, 
the Walsh transform $y= \widehat{x} =(y_0,y_1,\ldots,y_{2^{n}-1})$ of $x$ is defined by
\[
y_i = \sum_{j\in GF(2)^n} (-1)^{<i,j>} x_j,
\]
for any $n$-tuple binary vector $i$. 
We call $x_i$ (resp. $y_i$) the time-domain component (resp. transform-domain coefficient) of the signal with size $2^n$.
For basic properties and references on Walsh transforms, 
we refer to \cite{walsh-book2,my_new_submission}.

Let $f$ be a probability distribution of an $n$-bit random variable 
$\mathcal{X}=(X_n,X_{n-1},\ldots,X_1)$, where each $X_i\in \{0,1\}$.
Then, $\widehat{f}(m)$ is the \emph{bias} of the Boolean variable $<m,\mathcal{X}>$ for any fixed 
$n$-bit vector $m$,
which is often called the output \emph{pattern} or \emph{mask}. 
Here,
recall that a Boolean random variable $\mathcal{A}$ has \emph{bias} $\epsilon$, which is defined by
 $\epsilon=E[(-1)^\mathcal{A}] = \Pr(\mathcal{A}=0)-\Pr(\mathcal{A}=1)$.
Hence, if $\mathcal{A}$ is uniformly distributed, $\mathcal{A}$ has bias 0.
Obviously, the pattern $m$ should be nonzero.

Walsh transforms were used 
in statistics to find dependencies within a multi-variable data set.
In the multi-variable tests, each $X_i$ indicates the presence or absence (represented by `1' or `0') 
 of a particular feature in a pattern recognition experiment.
Fast Walsh Transform (FWT) is used to obtain all coefficients $\widehat{f}(m)$ in one shot.
By checking the Walsh coefficients one by one and identifying the large\footnote{We use the convention in signal processing to refer to the large transform-domain coefficient $d$ as the one
  with a large absolute value throughout the paper.} ones, 
we are able to tell the dependencies among $X_i$'s.

\section{Review on Shannon's Channel Coding Theorem}
\label{sect_shannon}

We briefly review Shannon's famous channel coding theorem (cf. \cite{it-book}).
First, we recall basic definitions of Shannon entropy.
The entropy $H(X)$ of a discrete random variable $X$ with alphabet
$\mathcal{X}$ and probability mass function $p(x)$
is defined by
\[
H(X)=-\sum_{x\in \mathcal{X}} p(x)\log_2 p(x).
\]
The joint entropy $H(X_1,\ldots,X_n)$ of a collection of discrete random variables $(X_1,\ldots,X_n)$
with a joint distribution $p(x_1, x_2, \ldots, x_n)$ is defined by
\[
H(X_1,\ldots,X_n)= - \sum_{x_1, x_2, \ldots, x_n} p(x_1, x_2, \ldots, x_n)\log_2 p(x_1, x_2, \ldots, x_n).
\]
Define the conditional entropy $H(Y|X)$ of a random variable $Y$ given another $X$ as
\[
H(Y|X)=\sum_x p(x)H(Y|X=x).
\]
The mutual information $I(X;Y)$ between
two random variables $X,Y$ is equal to $H(Y)-H(Y|X)$, which always equals $H(X)-H(X|Y)$.
 A communication channel is a system in which the output $Y$ depends
probabilistically on its input $X$. It is characterized by a probability
transition matrix that determines the conditional distribution of the
output given the input.

\begin{theorem}[Shannon's Channel Coding Theorem]
\label{thm_1}
Given a channel, denote the input, output by $X,Y$ respectively.
We can send information at the maximum rate $C$ 
bits per transmission with an arbitrarily low probability of error, where $C$ is the channel capacity defined by 
\begin{equation}\label{E_capacity_def}
C = \max_{p(x)} I(X;Y), 
\end{equation}
and the maximum is taken over all possible input distributions $p(x)$.
\end{theorem}

For the binary symmetric channel (BSC) with crossover probability\footnote{that is, 
the input symbols are complemented
with probability $p$}
 $p$, $C$ can be expressed by (cf. \cite{it-book}):
\begin{equation}\label{E_capacity}
C = 1 - H(p) 
\text{ bits/transmission.}
\end{equation}

Herein, we refer to the BSC with crossover probability $p = (1 + d)/2$ and $d$ is small (i.e., $|d|\ll 1$) as an extremal BSC.
We can prove for the channel capacity for an extremal BSC (see Appendix for proof):
\begin{corollary}[extremal BSC]\label{cor1}
Given a BSC channel with crossover probability
 $p = (1 + d)/2$, if $d$ is small (i.e., $|d|\ll 1$), then, 
$C\approx c_0\cdot d^2$, where the constant $c_0=1/(2\log 2)$.
\end{corollary}
Therefore, for an extremal BSC,
we can send one bit with an arbitrarily low probability of error with the minimum number of transmissions
 $1/C= (2\log 2)/{d^2}$, i.e., $O(1/d^2)$.
In next section, we will translate
Corollary \ref{cor1} to two useful statistical results. 
Interestingly, note that in communication theory, 
this extremal BSC is rare because of its low efficiency \cite{Shokrollahi} and 
we typically deal with $|d| \gg 0$.

\section{Statistical Translations of Shannon's Theorem}
\label{sect_stat_tran}

Let $X_0,X_1$ denote the Boolean random variable with bias $+d$, $-d$ respectively 
(and we restrict ourselves to $|d|\ll 1$). 
 Denote the probability distribution of $X_0$, $X_1$ by $D_0$, $D_1$ respectively.
Let $D\in \{D_0, D_1\}$.
We are given a binary sequence of random bits with length $N$, and each bit is independent and identically distributed (i.i.d.)
 following the distribution $D$. As a consequence of Shannon's channel coding theorem, 
we now solve a hypothesis testing problem in statistics:
 answer the minimum $N$ required to decide whether $D=D_0$ or $D=D_1$ with an arbitrarily low probability of error.

We translate this problem into a BSC channel coding problem as follows. 
The inputs are transmitted through a BSC with error probability $p=(1-d)/2$.
By Shannon's channel coding theorem, 
with a minimum number of $N=1/C$ transmissions, 
we can reliably (i.e., with an arbitrarily low probability of error) determine whether the input 
 is `0' or `1'. 
The former case implies that the received sequence corresponds to the distribution $D_0$ 
(i.e., a bit `1' occurs in the output sequence with probability $p$), while
the latter case implies that the received sequence corresponds to the distribution $D_1$
(i.e., a bit `0' occurs in the output sequence with probability $p$).
This solves the problem stated above.
Using Corollary \ref{cor1} with $p=(1-d)/2$ (for $|d|\ll 1$), 
we have $N=(2\log 2)/{d^2}$, i.e., $O(1/d^2)$.
Thus, we have just shown that Shannon's Channel Coding Theorem can be translated to solve
the following hypothesis testing problem:

\begin{theorem}\label{thm_2}
Assume that the boolean random variable $\mathcal{A}$, $\mathcal{B}$ has bias $+d$, $-d$ respectively and $d$ is small.
We are given a sequence of random samples, 
which are
i.i.d. following the distribution of either $\mathcal{A}$ or $\mathcal{B}$.
We can tell the sample source with an arbitrarily low probability of error,
 using the minimum number $N$ of samples $(2\log 2)/{d^2}$, i.e., $O(1/d^2)$.
\end{theorem}

Further, the following variant is 
more frequently encountered in hypothesis testing,
 in which we have to deal with a biased distribution and a uniform distribution altogether.

\begin{theorem}\label{thm_3}
Assume that the boolean random variable $\mathcal{A}$ has bias $d$ and $d$ is small. We are
given a sequence of random samples, 
which are 
i.i.d. following the distribution of either $\mathcal{A}$ or a uniform distribution. 
We can tell the sample source with an arbitrarily low probability of error,
 using the minimum number $N$ of samples $(8\log 2)/{d^2}$, i.e.,
$O(1/d^2)$.
\end{theorem}

\begin{proof}

It is clear that the construction of using a BSC
 in the proof of Theorem \ref{thm_2}
does not work here,
as the biases (i.e., $d,0$ respectively) of the two sources are non-symmetric. Thus,
we propose to use Shannon's channel coding theorem with a non-symmetric binary channel rather than a BSC.

Assume the channel with the following transition matrix
\[
p(y|x)=\left(
\begin{array}{cc}
1-p_e & p_e \\
1/2 & 1/2
\end{array}
\right),
\]
where $p_e=(1-d)/2$ and $d$ is small.
The matrix entry in the $x$th row and the $y$th column denotes the conditional
probability that $y$ is received when $x$ is sent.
So, the input bit $0$ is transmitted by this channel with error probability $p_e$
(i.e., the received sequence has bias $d$ if input symbols are 0)
 and the input bit $1$ is transmitted with error probability 1/2
(i.e., the received sequence has bias $0$ if input symbols are 1).

To compute the channel capacity $C$ (i.e., to find the maximum)
 defined in (\ref{E_capacity_def}),
 no closed form solution exist in general. Nonlinear optimization algorithms 
\cite{Arimoto1972channel_capacity,Blahut1972channel_capacity}
 are known to find a numerical solution. Below, 
we propose a simple method to give a closed form estimate $C$ for our extremal binary channel.
As $I(X;Y)= H(Y)-H(Y|X)$, we first compute $H(Y)$ by
\begin{equation}\label{E_tmp1}
H(Y) = H\Bigl( p_0(1-p_e)+ (1-p_0)\times \frac{1}{2} \Bigr),
\end{equation}
where $p_0$ denote $p(x=0)$ for short.
Next, we compute $H(Y|X)$ as follows,
\begin{eqnarray}
H(Y|X) &=& \sum_x p(x)H(Y|X=x) \nonumber\\
&=& p_0\Bigl( H(p_e)-1 \Bigr) +1.\label{E_tmp2}
\end{eqnarray}
Combining (\ref{E_tmp1}) and (\ref{E_tmp2}), we have
\[
I(X;Y) = H\Bigl( p_0\times \frac{1}{2} - p_0p_e + \frac{1}{2} \Bigr) - 
  p_0H(p_e) + p_0 -1.
\]
As $p_e=(1-d)/2$, we have
\[
I(X;Y)
=H(\frac{1 + p_0d}{2}) -p_0\Bigl( H(\frac{1-d}{2})-1 \Bigr) - 1.
\] 
We apply (\ref{approx_H_special}) (in Appendix)
\begin{equation}\label{E_tmp3}
I(X;Y) = -\, \frac{p_0^2 d^2}{2\log 2} - p_0\Bigl( H(\frac{1-d}{2})-1 \Bigr) + O(p_0^4 d^4),
\end{equation}
for small $d$. Note that
 the last term $O(p_0^4 d^4)$ on the right side of (\ref{E_tmp3}) is ignorable.
 Thus, $I(X;Y)$ approaches the maximum when
\[
p_0 = -\,\frac{ H(\frac{1-d}{2})-1 }{d^2/(\log 2)} \approx \frac{d^2/(2\log 2)}{d^2/(\log 2)}=\frac{1}{2}.
\]
Consequently, we estimate the channel capacity from (\ref{E_tmp3}) by
\[
C\approx -\,\frac{1}{4}d^2/(2\log 2) +\frac{1}{2} \Bigl(1-H(\frac{1-d}{2})\Bigr) \approx -d^2/(8\log 2) + d^2/(4\log 2),
\]
which is $d^2/(8\log 2)$.

\end{proof}

\begin{remark}\label{remark1}
In statistical cryptanalysis (cf. \cite{vaudenay_ccs96,vaudenay_textbook2006}),
Theorem \ref{thm_2} and Theorem \ref{thm_3} were known in slightly different contexts:
 the probability of error is a parameter and the sample number is known on the order of $1/{d^2}$.
By asking for an arbitrarily low probability of error,
we are able to give an alternative proof using channel capacity rather than relative entropy (or Kullback-Leibler distance).
While the latter is used as the classical tool to solve hypothesis testing problems,
here we show that hypothesis testing problems can be linked to channel capacity. 
\end{remark}

\section{Sampling Theorems with Incomplete Signals}
\label{sect_main_result}

In this section, 
we apply the hypothesis testing result (Theorem \ref{thm_3}) to
two sampling problems (the classical and generic versions).
Without loss of generality,
we assume the discrete statistical signals are not restricted to a particular application domain.
Assume that (possibly noise-corrupted) signals are $2^n$-valued and noises are uniformly distributed.
For the signal detection problem (i.e., to test presence of real signal),
we adopt the conventional approach of
 statistical hypothesis testing.
Rather than using the direct signal detection method (as done in specific application domains),
we propose to perform the test between
the associated distribution and the uniform distribution.

We give the mathematical model on the signal $F$ as follows.
$F$ is an arbitrary (and not necessarily deterministic) function.
Let $X$ be the $n$-bit output sample of $F$, 
assuming that the input is random and uniformly distributed.
Denote the output distribution of $X$ by $f$.
%
%
Note that our assumption on a general setting of
 discrete statistical signals
is described by the assumption that $F$ is an arbitrary yet fixed function.

Firstly, 
the classical sampling problem (which can be interpreted as 
the classical distinguisher\footnote{As mentioned in Remark \ref{remark1}, the problem
statement of the classical distinguisher is slightly different; it 
often deals with a large $d$ (using a slightly different $N$) rather than the largest $d$ (cf. \cite{vaudenay_textbook2006}).})
 is formally stated as follows.

\begin{theorem}[Classical Sampling Problem]\label{thm_5}
Assume that the largest Walsh coefficient of $f$ is $d=\widehat{f}(m_0)$ for a 
nonzero $n$-bit vector $m_0$.
We can detect $F$ with an arbitrarily low probability of error, using minimum number $N=(8\log 2)/{d^2}$ of samples of $F$,
i.e., $O(1/{d^2})$.
\end{theorem}

%
The proof can be easily obtained by applying Theorem \ref{thm_3} and we omit it here.
The classical sampling problem
 assumes that $F$ together with the its characteristics (i.e., the largest Walsh coefficient $d$) are known \emph{a priori}.
It aims at detecting signal with an arbitrarily low probability of error, using minimum samples.

Next, we will present our main sampling theorem, a more practical (and widely applicable) sampling theorem formally.
Assuming that it is infeasible to know signal $F$ \emph{a priori}, 
we want to detect signals with an arbitrarily low probability of error and with bounded sample size.
Note that the sampled signal is incomplete (possibly noisy) and the associated distribution is noisy (i.e., not precise).
And we call this problem as generic sampling with incomplete noisy signals.
In contrast to the classical distinguisher, this result can be interpreted as a generalized distinguisher\footnote{With $n=1$,
this appears as an informal
result in cryptanalysis, which is used as a black-box analysis tool in several crypto-systems.}  
 in the context of statistical cryptanalysis.
We give our first result with $n=1$ below.

\begin{theorem}[Generic Sampling Problem with $n=1$]\label{cor_thm_5}
Assume that the sample size of $F$ is upper-bounded by $N$.
Regardless of the input size of $F$,
in order to detect $F$ with an arbitrarily low probability of error, 
 it is necessary and sufficient to have the following condition satisfied,
i.e.,
 $f$ has a nontrivial Walsh coefficient $d$ with 
$|d|\ge c/{\sqrt{N}}$, where the constant $c=\sqrt{8\log 2}$.
\end{theorem}

\begin{proof}

Note that the only nontrivial Walsh coefficient $d$ for $n=1$ is $\widehat{f}(1)$, 
which is nothing but the bias of $F$.
First, we will show by contradiction that this is a necessary condition.
That is, if we can identify $F$ with an arbitrarily low probability of error, 
then, we must have $|d|\ge c/{\sqrt{N}}$.
Suppose $|d| < c/{\sqrt{N}}$ otherwise.
Following the proof of Theorem \ref{thm_3}, we know that the error probability is bounded away from zero 
as the consequence of Shannon's Channel Coding Theorem. This is contradictory.
Thus, we have shown that the condition on $d$ is a necessary condition.
Next, we will show that it is also a sufficient condition.
That is, if $|d|\ge c/{\sqrt{N}}$, then, 
we can identify $F$ with an arbitrarily low probability of error.
This follows directly from Theorem \ref{thm_5} with $n=1$.
We complete our proof. 
\end{proof}
Now, we make a generalized proposition for $n\ge 1$, which incorporates Theorem \ref{cor_thm_5} as a special case:

\begin{proposition}[Generic Sampling Problem with $n\ge 1$]\label{cor_thm_5b}
Assume that the sample size of $F$ is upper-bounded by $N$.
Regardless of the input size of $F$,
in order to detect $F$ with an arbitrarily low probability of error, 
 it is necessary and sufficient to have the following condition satisfied,
i.e.,
$\sum_{i\ne 0} (\widehat{f}(i))^2 \ge (8n\log 2)/N$.
\end{proposition}
%
We note that the sufficient condition can be proved based on results of classic distinguisher (i.e., Squared Euclidean Imbalance)
 which uses the notion of relative distance and states that $\sum_{i\ne 0} (\widehat{f}(i))^2 \ge (4n\log 2)/N$ is required for high probability \cite{vaudenay_textbook2006}.

According to Theorem \ref{cor_thm_5} and Proposition \ref{cor_thm_5b},
 note that a real signal $F$ should have the following property in the form of $\ell_2$ norm of the associated distribution
 given the sample size $N$:

\[
\|\widehat{f}\|_2^2 \ge 1 + {8n\log 2 /N},
\]
where 
 the $\ell_2$ norm of $f$ is defined as
\[
\| f \|_2 = \sqrt{\sum_{i\in GF(2)^n} f(i)^2}.
\]

By duality of time-domain and transform-domain signals,
we make another proposition following Proposition \ref{cor_thm_5b}:

\begin{proposition}\label{prop}
The discrete statistical signals can be characterized 
by large Walsh coefficients of the associated distribution.
\end{proposition}

Proposition \ref{prop} implies that 
the most significant transform-domain signals are the largest coefficients in our generalized model.
This is a known fact in application domains such as images, voices etc.
Nonetheless, for those signals,
Walsh transform is directly applied to the time-domain samples 
rather than the associated distribution of the collected samples in our model;
in analogy to Proposition \ref{prop},
 it is known that those signals can be characterized by large Walsh coefficients as well.

\subsection{Cryptographic Significance on Sparse Walsh Transforms}

In symmetric cryptanalysis,
Walsh transforms play an essential role 
 (cf. \cite{Euro94Chabaud-Vaudenay,my_new_submission}), including bias computing.

Following the recent successful development of compressive sensing \cite{cs2006paper},
it is shown that surprisingly,
 sparse Fourier transform significantly outperforms FFT (Fast Fourier Transforms).
For the problem size $N$, $k$-sparse Fourier transforms ($k\ll N$) aims at faster computing
$k$ non-zero or large coefficients and $(N-k)$ zero or negligible small ones,
in comparison to FFT.
For instance, according to \cite[Fig. 1]{parallel_sfft}, 
with $N=2^{28}, k=50$, theoretical estimate on the time complexity of FFT is $N\cdot \log_2 N \approx 7\times 10^9$ units;
for sparse Fourier transforms, the estimated theoretical complexity is $10^7$ units, 
i.e., a great reduction factor of $700$ is obtained.

Due to the similarity of Fourier transform and Walsh transform, most recently,
research on sparse Walsh transform follows \cite{berkeley_sparse_wht_project,epfl_sparse_wht_project}.
As illustration, 
assume $k$ non-zero coefficients and $(N-k)$ zero coefficients in a simplified model.
With the same parameters ($N=2^{28}, k=50$) as above,
for sparse Walsh transform, 
the conservative theoretical time complexity\footnote{The required time-domain components for access is around $6700$ 
(see \cite[Theorem 1]{epfl_sparse_wht_paper2013}) rather than $N$ for FWT.}
 is around $38000$ units.
%
%
This time unit is not comparable to the one in the case of FWT, i.e., $7\times 10^9$ units.
Nonetheless, we estimate a rough reduction factor of 8000 by \cite[Fig. 8]{epfl_sparse_wht_paper2013}.
Additionally, for $k=2, 4, 12, 25$, sparse Walsh transform \cite{epfl_sparse_wht_paper2013} 
has the estimated time of $1600, 3000, 8200, 16400$ units respectively.

According to our discussions in this section,
it is natural to link the first key challenge to the generic approach of sparse Walsh transforms.
In \cite{crypto2004me,JoC2008}, 
 finding the largest Walsh coefficient is linked to maximum likelihood decoding problem for linear codes,
which is known to be NP-complete.
Assume $k$ large coefficients and $(N-k)$ zero or negligible small ones in a general setting.
It seems other than FWT, no efficient algorithms exist to compute sparse Walsh transforms.
In contrast, in the simplified $k$-sparse model, 
theoretical estimate for the time complexity corresponding to $k=1, 2$ is $(\log_2 N)^2, 2(\log_2 N)^2$.
That is, we have the complexity on the order of $(\log_2 N)^2$ (resp. $N\log_2 N$) 
in the simplified model (resp. the general model).
And we are working on approximate signal recovery in presence of noise to gain more insights about the first challenge.

\section{Concluding Remarks}
\label{sect_end}

We model general discrete statistical signals as the output samples of
an unknown arbitrary yet fixed function (which is the signal source).
We translate Shannon's channel coding theorem in the extremal case of a binary channel
 to solve a hypothesis testing problem.
Due to high probability of transmission error,
 this extremal binary channel is rare in communication theory.
Nonetheless, the translated result allows to solve a generic sampling problem, 
for which we know nothing about the signal source \emph{a priori} 
and we can only afford bounded sampling measurements.
Our main results demonstrate that the classical signal processing tool of Walsh transform is essential: 
 it is the large Walsh coefficient(s) that characterize(s) discrete statistical signals, 
regardless of the signal sources. 
By Shannon's theorem, 
we establish the \emph{necessary and sufficient} condition 
for the generic sampling problem under the general assumption of statistical signal sources.
It shows strong connection
between Shannon's theorem and Walsh transform;
both are the key innovative technologies in digital signal processing.
Our results can also be seen as generalization of the classic distinguisher;
 the latter is based on relative distance
and is the standard tool for statistical hypothesis testing problems. 
Finally, based on our preliminary work on sparse Walsh transforms in the context of compressive sensing, 
we discuss the cryptographic significance.
%

\section*{Appendix: Proof of Corollary \ref{cor1}}

Let $p=(1+d)/2$ and so $|d|\le 1$. For $|d|<1$,
we will first show  
\begin{equation}
\label{approx_H}
H\Bigl(\frac{1+d}{2}\Bigr) = 1 - \Bigl(\frac{d^2}{2} + \frac{d^4}{12} + \frac{d^6}{30} + \frac{d^8}{56} + 
\underbrace{\cdots}_{O(d^{10})}
 \Bigr) 
\times \frac{1}{\log 2}.
\end{equation}
We have
\begin{eqnarray}
- H\Bigl(\frac{1+d}{2}\Bigr) &=& \frac{1+d}{2}\log_2 \frac{1+d}{2} + \frac{1-d}{2}\log_2 \frac{1-d}{2} \\
&=& \frac{1}{\log 2} \Bigl( \frac{1+d}{2}\log \frac{1+d}{2} + \frac{1-d}{2}\log \frac{1-d}{2} \Bigr) \\
&=& \frac{1}{\log 2} \Bigl( \frac{1+d}{2}\log (1+d) + \frac{1-d}{2}\log (1-d) - \log 2 \Bigr) \\
&=& \frac{1}{\log 2} \Bigl( \frac{1}{2}\log (1-d^2) + \frac{d}{2}\log \frac{1+d}{1-d} - \log 2 \Bigr) \label{E_tmp_a0}
\end{eqnarray}
by definition of entropy.
Using Taylor expansion series for $0\le d<1$, we have
\begin{eqnarray}
\log(1-d^2) &=& - \Bigl(d^2 + \frac{d^4}{2} + \frac{d^6}{3} + \frac{d^8}{4} + \cdots \Bigr) \label{E_tmp_a1}\\
\log \frac{1+d}{1-d} &=& 2 \Bigl( d + \frac{d^3}{3} + \frac{d^5}{5} + \frac{d^7}{7} + \cdots \Bigr) \label{E_tmp_a2}
\end{eqnarray}
 Putting (\ref{E_tmp_a1}) and (\ref{E_tmp_a2}) into (\ref{E_tmp_a0}), we have
\begin{eqnarray*}
- H\Bigl(\frac{1+d}{2}\Bigr) &=& \frac{1}{\log 2} \Bigl(-\,\frac{1}{2} \Bigl(d^2+ \frac{d^4}{2} + \frac{d^6}{3} + \frac{d^8}{4} + \cdots \Bigr) + \\
 && \Bigl( d^2 + \frac{d^4}{3} + \frac{d^6}{5} + \frac{d^8}{7} + \cdots \Bigr) - \log 2 \Bigr) \\
&=& \frac{1}{\log 2} \Bigl( \frac{d^2}{2} + \frac{d^4}{12} + \frac{d^6}{30} + \frac{d^8}{56} + \cdots \Bigr) -1,
\end{eqnarray*}
which leads to (\ref{approx_H}) for $0\le d<1$. 
For $-1<d\le 0$, we use symmetry of entropy $H(\frac{1+d}{2})=H(\frac{1-d}{2})$ and 
apply above result to justify the validity of (\ref{approx_H}) for $|d|<1$.

Note that
if $|d|\ll 1$, (\ref{approx_H}) reduces to
\begin{equation}\label{approx_H_special}
H\Bigl(\frac{1+d}{2}\Bigr) = 1 - d^2 /(2\log 2) + O(d^4).
\end{equation}
So, we can calculate $C$ in (\ref{E_capacity}) by
\[
C = 1 - H\Bigl( \frac{1+d}{2} \Bigr) 
= \Bigl(d^2 + O(d^4)\Bigr)/(2\log 2) 
 \approx \frac{d^2}{2\log 2},
\]
which completes our proof.

\end{document}